\documentclass[a4paper, 12pt]{article}
\usepackage[a4paper, total={6.5in, 9in}]{geometry}
\usepackage[english]{babel}
\usepackage{natbib} % has a nice set of citation styles and commands
\usepackage{siunitx} % for proper typesetting of numbers and units
\usepackage{tikz} % nice language for creating drawings and diagrams
\usetikzlibrary{arrows}
\usepackage{amsmath,amssymb}
\usepackage{authblk}

\newcommand{\xhdr}[1]{\vspace{0.05in}\noindent \textbf{#1.} \noindent}
\renewcommand{\b}[1]{\left[#1\right]}
\newcommand{\E}[1]{\mathbb{E}\b{#1}}
\usetikzlibrary{positioning,decorations.pathmorphing,shapes,decorations.pathreplacing}
\newcommand{\p}[1]{\left(#1\right)}
\def\squarebox#1{\hbox to #1{\hfill\vbox to #1{\vfill}}}
\newcommand{\qed}{\hspace*{\fill}\vbox{\hrule\hbox{\vrule\squarebox{.667em}\vrule}\hrule}\smallskip}
\newenvironment{proof}{\noindent{\bf Proof:~~}}{\(\qed\)}

\newcommand{\eps}{\varepsilon}

\newcommand{\comment}[1]{}

\newtheorem{theorem}{Theorem}[section]
\newtheorem{corollary}[theorem]{Corollary} 
\newtheorem{lemma}[theorem]{Lemma}
\newtheorem{proposition}[theorem]{Proposition}

\newtheorem{claim}[theorem]{Claim}

\title{Stochastic Model for Sunk Cost Bias \thanks{Work supported in part by BSF grant 2018206, Vannevar Bush Faculty Fellowship, MURI grant W911NF-19-0217, AFOSR grant FA9550-19-1-0183 and ISF grant 2167/19.}%\thanks{Missing proofs are included as supplementary material.}
}

 \author[1]{\small Jon Kleinberg}
 \author[2]{Sigal Oren}
 \author[1]{Manish Raghavan}
 \author[2]{Nadav Sklar}
 \affil[1]{%
	Computer Science Dept.\\
	Cornell University\\
	Ithaca, New York, USA
}
\affil[2]{%
	Computer Science Dept.\\
	Ben-Gurion University of the Negev\\
	Be'er Sheva, Israel
}

\begin{document}
\maketitle

\begin{abstract}
	
We present a novel model for capturing the behavior of an agent exhibiting sunk-cost bias in a stochastic environment. Agents exhibiting sunk-cost bias take into account the effort they have already spent on an endeavor when they evaluate whether to continue or abandon it. We model planning tasks in which an agent with this type of bias tries to reach a designated goal. Our model structures this problem as a type of Markov decision process: loosely speaking, the agent traverses a directed acyclic graph with probabilistic transitions, paying costs for its actions as it tries to reach a target node containing a specified reward. The agent's sunk cost bias is modeled by a cost that it incurs for abandoning the traversal: if the agent decides to stop traversing the graph, it incurs a cost of $\lambda \cdot C_{sunk}$, where ${\lambda \geq 0}$ is a parameter that captures the extent of the bias and $C_{sunk}$ is the sum of costs already invested. 

We analyze the behavior of two types of agents: naive agents that are unaware of their bias, and sophisticated agents that are aware of it. Since optimal (bias-free) behavior in this problem can involve abandoning the traversal before reaching the goal, the bias exhibited by these types of agents can result in sub-optimal behavior by shifting their decisions about abandonment. We show that in contrast to optimal agents, it is computationally hard to compute the optimal policy for a sophisticated agent. Our main results quantify the loss exhibited by these two types of agents with respect to an optimal agent. We present both general and topology-specific bounds.

\end{abstract}

\newpage

\section{Introduction} \label{sec-sunk} Imagine that you paid \$50 to go to a rock concert and five minutes into the show you realize that the acoustics are horrible, the venue is smelly and the band is not playing well. Will you stay or go? Would you have made a different decision if the concert were free? Many will choose to stay in the concert in the first case but leave in the second one. This phenomenon, in which effort or cost invested in the past affects current decisions, has fascinated many researchers from different disciplines. This is evident from the variety of names the phenomenon has been studied under:
%This phenomenon received different names including given including:  
the sunk cost effect \citep{arkes1985psychology, thaler1980toward}, escalation of commitment \citep{staw1976knee} and the Concorde fallacy \citep{dawkin1976selfish, weatherhead1979savannah}. The latter is named after the famous supersonic airplane whose development was continued long after it was clear that it did not have any economic justification. Some of the many situations in which sunk cost has been observed include auctions \citep{augenblick2016sunk}, medical treatment \citep{coleman2010sunk,eisenberg2012falling}, project development \citep{garland1990throwing} and poker \citep{smith2009poker}.  

Factoring sunk cost into future decisions is at odds with standard economic theory advocating that decisions should only depend on marginal costs and gains. Several explanations have been offered for the sunk cost effect. \citep{arkes1985psychology} suggest it is a manifestation of the ``do not waste'' rule that we are often taught as children. Early work in psychology \citep{aronson68dissonance, staw1980rationality} attributes this to self justification: decision-makers continue in the same course of action to justify their initial decision and avoid cognitive dissonance. \citep{thaler1980toward} applies prospect theory \citep{kahneman-prospect} to explain this bias. %(in particular loss aversion together with a fixed reference point) 

%Perhaps move to realted work
\comment{
Heath (\cite{heath1995escalation, heath1996mental}) suggests a more general model that also captures situations in which individuals invest less than they should have. Lastly, a recent paper (\cite{haita2013sunk}) relates sunk cost bias to realization bias (\cite{barberis2012realization}). 
}

In this paper, we analyze the performance of agents who engage in activities that require multi-step planning in the presence of sunk cost bias. Through this, our work is situated in a recently growing literature in algorithmic game theory aiming to model and theoretically analyze planning related biases (e.g., \citep{kleinbergTime, Gravin:present-bias, kleinberg-soph, Kraft-time,present-bias-tang, Kleinberg-mult-bias}). Despite the crucial role played by sunk cost bias in empirical studies of behavior, it has received very little theoretical study in this style; the main prior contribution is a model of \citep{Kleinberg-mult-bias} that considered the interplay of sunk cost bias with present bias in a deterministic setting. However, looking at the scenarios discussed so far, we see that many of them crucially involve agents who are planning with respect to uncertainty about future outcomes: the sunk cost bias often becomes particularly dangerous when an agent takes an action while the future remains uncertain, and then is subject to the sunk cost from this action after the uncertainty is resolved. Indeed, many of the most natural questions that arise when studying sunk cost bias in isolation (separately from other effects such as present bias) do not have natural formulations in deterministic models. It is therefore an important and unexplored question to analyze the effects of sunk cost bias in a model featuring uncertainty.

%Our paper is situated in a recently growing literature in algorithmic game theory aiming to model and theoretically analyze planning related biases (e.g., \citep{kleinbergTime, Gravin:present-bias, kleinberg-soph, Kraft-time,present-bias-tang, Kleinberg-mult-bias}). This literature considered the modeling of sunk cost bias in a deterministic setting: \citet{Kleinberg-mult-bias} presented and analyzed a graph theoretic model featuring the interplay between sunk cost bias and present bias. However, looking at the scenarios discussed so far, we see that many of them crucially involve agents who are planning with respect to uncertainty about future outcomes: the sunk cost bias often becomes particularly dangerous when an agent takes an action while the future remains uncertain, and then is subject to the sunk cost from this action after the uncertainty is resolved. Another important aspect of studying sunk cost in a stochastic environment is that it allows us to study it in isolation (as opposed to the interplay between sunk cost bias and present bias which was studied by \citet{Kleinberg-mult-bias}. It is a natural question, therefore, to analyze the effects of sunk cost bias in a model featuring uncertainty.

\xhdr{A Stochastic Model of Sunk Cost} %\label{sec-stochastic-sunk}
We present a stochastic model aiming to study scenarios involving sunk cost bias and uncertainty. We focus on situations in which taking sunk cost into account is irrational,\footnote{We note that as \citep{mcafee2010sunk} advocates taking sunk
	cost into account can sometimes be rational. For example, in a project with an
	unknown completion time, the time already invested can hint at the actual
	completion time.}  for example, a gambler who has already ``invested'' \$100 in a slot machine and keeps playing because they are sure that after ``investing'' so much money they will hit the jackpot soon. These situations involve the following basic ingredients: an agent needs to formulate a plan in which it traverses a set of states, trying to reach a designated goal state. The transitions between states are stochastic based on the agents' actions, and the agent must deal with its own sunk cost bias as it formulates and updates its plan for traversing the states.

Motivated by these considerations, we model the agent's problem using a directed acyclic graph in which the agent must traverse a path from a start node $s$ to a target node $t$, with a reward of $R$ for reaching $t$. Each node is assigned a cost of going forward and there is a probability distribution on its outgoing edges, determining the next node that the agent would reach. This is a type of a Markov decision process. After each step, the agent has to choose whether to stop or keep traversing the graph. If the agent at some node $v$ decides to continue traversing the graph then it pays the cost assigned to $v$ and moves to a neighboring node of $v$ determined stochastically according to a distribution on $v$'s neighbors.

We model the decision making process of agents with sunk cost bias similarly to \citep{Kleinberg-mult-bias}. We assume that an agent exhibiting sunk cost bias has some parameter {$\lambda\geq 0$}  that represents the extent to which the agent cares about sunk cost\footnote{It is natural to limit $\lambda$ to non-negative values as negative values imply that the agent believes that if it would stop it would get some of its investment back. This is the opposite of sunk cost bias.}. {While the regime of $0\leq \lambda \leq 1$ is perhaps more natural, for sake of generalization, we present and analyze our model using the more general assumption of $\lambda\geq 0$.} Let $C_{sunk}$ be the cost that the agent already invested. An agent with sunk cost bias views the option of quitting as having a cost of $\lambda C_{sunk}$, hence it will continue traversing the graph if and only if the expected payoff from continuing is greater than $-\lambda C_{sunk}$. We stress that, while our paper uses the basic means of accounting for sunk cost employed by \citep{Kleinberg-mult-bias}, the models studied in the two papers are inherently different. The current paper studies a stochastic model focusing on the effects of sunk cost bias by itself; in contrast, the model of \citep{Kleinberg-mult-bias} is a deterministic formalism that studies the simultaneous effect of present bias and sunk cost bias, and through its deterministic structure cannot encapsulate the key issues that we address here.

%$C_{comp}$ be his expected cost for reaching $t$ and $R$ a reward for reaching $t$. An agent with sunk cost bias views the option of quitting as having a value of $-\lambda C_{sunk}$, hence it will 
%continue traversing the graph as long as $R- C_{comp} \geq  -\lambda C_{sunk}$.
{Following O'Donoghue and Rabin \citep{odonoghue-now-or-later,odonoghue-choice-procrastination},
we analyze the behavior of two types of agents: naive agents that are unaware of their bias, and sophisticated agents that are aware of it. Even though a sophisticated agent is aware of its bias it cannot simply ignore it. However, loosely speaking, it can take future actions to minimize the negative implications of its bias. 
}
To understand the behavior of the different agents, it is best to walk through a simple example. Consider a slot machine with a probability of $1/3$ of winning a reward of \$10.  As part of a promotion the
casino prices the first round at \$3 instead of the usual price of \$4.
% first round costs \$3 and the second round costs \$4. 
This scenario is depicted in the graph in Figure \ref{fig-slot-machine}. An unbiased (e.g., optimal) agent would only play as long as the expected payoff is greater than $0$. Hence, in this game, it will only play the first round. Now, consider a naive agent exhibiting sunk cost bias with a parameter $\lambda = 1/3$. If it loses the first round it would have a sunk cost of $1/3 \cdot 3=1$. Thus, its payoff for quitting would be $-1$ while its expected payoff for continuing would be $1/3\cdot 10 - 4 = -2/3$. The naive agent will therefore play the next round as well and attain a negative expected payoff. In fact, we show that the negative payoff of a naive agent can be exponentially large in the size of the graph.

% \socomment{TODO: Make the edges of the triangle shorter}
	\begin{figure}[t]
		\centering
		\begin{tikzpicture}[->,shorten >=1pt,auto,node distance=2cm, thin]
		\node(0) [circle, draw, inner sep=0pt, minimum size=1cm] at (0,0) {$s;3$};
		\node (1) [circle, draw, inner sep=0pt, minimum size=1cm] at (1.5,1.5) {$v_1;4$};
		\node (2) [circle, draw, inner sep=0pt, minimum size=1cm] at (4,1.5) {$v_2$};
		\node (5) [circle, draw, inner sep=0pt, minimum size=1cm] at (3,0) {$t$};
		
		\path[every node/.style={sloped,anchor=south,auto=false}]
		(0) edge node {{ \small$2/3$}} (1)
		(1) edge node {{ \small$2/3$}} (2)

		(0) edge node {{ \small$1/3$}} (5)
		(1) edge node {\small$1/3$} (5)
		;
		\end{tikzpicture}
		\caption{For $R=10$, naive sunk cost bias agents will continue at $v_1$ and end up with a negative expected payoff.  }
		\label{fig-slot-machine}
	\end{figure}
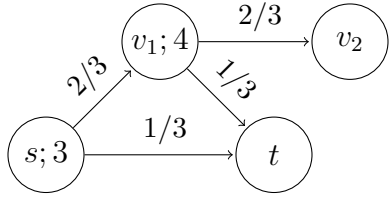

A sophisticated agent with sunk cost bias is aware of its bias. 
{
This means that it knows which action it will take in any 
subsequent node for any possible sunk-cost and can use this information to compute its expected payoff. This is different than a naive agent who is unaware of its bias and hence, wrongfully, believes that in the future, it will behave the same as an unbiased agent and therefore, in all subsequent nodes, will have the same expected payoff as an unbiased agent.

We observe that the expected payoff of a sophisticated agent is always non-negative since it would stop traversing the graph if it knows that its expected payoff for doing so will be negative. This is the case in the example in Figure \ref{fig-slot-machine} in which the sophisticated agent knows that if it will play the first round and lose than it will also play the second round. This implies that its expected payoff for playing the slot machine is $1/3\cdot 10 -3 + 2/3(1/3\cdot 10 -4) = -1/9.$
}
% and hence would stop traversing the graph if it knows that as a result of its bias it will end up with a negative expected payoff. Therefore, the payoff of a sophisticated agent will always be non-negative. 
%In the example in Figure \ref{fig-slot-machine}, a sophisticated agent %will not traverse the graph at all. 
The fact that sophisticated agents sometimes stop traversing the graph prematurely makes their payoff potentially smaller than that of an optimal agent, although they avoid some of the more dramatic payoff shortfalls of naive agents.

\xhdr{Results} 
We begin by considering naive agents. Since they believe that in the future they will behave as optimal agents, their policy can be efficiently computed similarly to the policy of optimal agents. This is done by going over the graph in reverse topological order
and computing the expected payoff of continuing at each node. Naive agents can then decide whether to continue or not by comparing these values against their sunk cost. We show that since they are oblivious to their sunk cost bias, we can construct instances in which they accumulate some sunk cost in the beginning. Then, due to this sunk cost, they continue traversing the graph and accumulate more and more sunk cost even when their expected payoff is negative. As a result, they may end up with a negative payoff that is exponential in the graph's size. This result illustrates the danger of marketing strategies that reduce initial entrance costs to lure individuals to begin {some risky endeavor (e.g., as the investing app Robinhood gives a free stock to anyone opening a new account) or take on some bad habit (e.g., tobacco companies giving free cigarettes to employees)}. 

Our main focus in this paper is on sophisticated agents. The behavior of agents that are aware of their bias is much more complex. In contrast to optimal and naive agents, they cannot compute their optimal policy by going over the nodes in reverse topological order. This is because the decision of whether to stop or continue at each node depends on the amount of sunk cost they accumulated along the way. When there are different paths reaching the same node, the amount of sunk cost may vary depending on the realized path. In fact, we show that the problem of computing the optimal policy for a sophisticated agent is \#P-Hard. This is done by reducing from the $0-1$ knapsack solution counting problem. Roughly speaking, we construct instances in which computing the expected payoff of a sophisticated agent if it starts traversing the graph requires counting the number of valid solutions to a corresponding knapsack problem. It is worth noting that a different type of hardness (i.e., NP-hardness) was proven by \citep{Kleinberg-mult-bias} for sophisticated agents exhibiting both sunk cost bias and present bias. The results strengthen one another and show that under different models being sophisticated about one's sunk cost bias may be quite challenging. As part of future research, it would be fascinating to model and analyze heuristics that individuals may use to bypass this hardness.

We continue with comparing the payoff of a sophisticated agent against the payoff of an optimal agent. Roughly speaking, sophisticated agents exhibit the opposite problem than naive agents: they take a too conservative approach and stop traversing the graph prematurely. As a result, they can have a payoff of $0$ even when the optimal agent has a positive payoff. {When $\lambda$ is approaching infinity this gap can attain its maximal value which is $R$. However, the payoff difference of $R$ is far from tight for smaller values of $\lambda$. Hence, we look for tighter bounds that are more suitable for such values.} We provide a number of bounds on the difference between the payoff of optimal and sophisticated agents. For example, we show that $\pi_s \geq \pi_o - \frac{\lambda}{1+\lambda} \cdot R$, where $\pi_s$ and $\pi_o$ are the expected payoffs of the sophisticated and optimal agents respectively.
%This bound is tight only for a slightly different model in which each edge may have a different cost.
 %However, by analyzing a small $3$-node graph, we show that this is not a tight bound for our model. 
We present some evidence that this bound is not tight, particularly, for $0\leq \lambda \leq 1$. We suspect that graphs achieving the worst case difference, for $0\leq \lambda \leq 1$, are fan graphs (graphs that include a path plus an edge from each node in the path to the target). We show that for such graphs $\pi_s \geq \pi_o -\frac{1}{e} \cdot \lambda \cdot R$ and prove that this bound is {essentially} asymptotically tight (in the graph's size).

\section{Model and Naive Agents} In our model an agent is traversing a directed acyclic graph (i.e., DAG).
The graph is a Markov decision process (i.e., MDP) where each state $u$ has a cost $c(u)$ \footnote{This is a restricted type of MDP in which for every node $u$ the transition cost to each neighbor is the same.} which is the cost of an agent at $u$ to continue traversing the graph and for each neighboring states $u$ and $v$,
$p(u,v)$ denotes the probability of a $u \to v$ transition. The graph also has a designated target node $t$. If the agent reaches $t$ it receives a 
reward of $R$. An agent traversing the graph forms a policy that decides for each node in the graph 
whether to continue traversing the graph or not. The goal of an agent is to choose a policy that 
maximizes its \emph{expected payoff} -- the probability of reaching the target multiplied by $R$ minus
the expected cost. We can define the expected payoff of an agent inductively as follows: if the agent decides to continue from $u$, the
expected payoff at a vertex $u$ is the weighted average of the expected
payoff of each neighbor vertex minus the cost for continuing. If it decides to stop, its expected payoff is $0$. 
We denote the expected payoff of an optimal agent (i.e., bias-free agent) currently at node $u$ by $\pi_o(u)$. We have that:
$${
\pi_o(u) = \max\{\sum_{v \in N(u)} p(u,v) \cdot \pi_o(v) - c(u), 0\} }$$
where $N(u)$ denotes the set of neighbors of node $u$.

An agent with sunk cost bias is characterized by a parameter ${\lambda \geq 0}$ that captures the intensity of its bias. 
%\nsedit{Note that the regime in which $\lambda \leq 1$ which implies that the sunk cost is at most the actual cost is quite natural. However, higher values of $\lambda$ may correspond to situations in which an individual has such a strong distaste for waste that abandoning entails a cost higher than what he invested. Therefore, we analyze our model for any $\lambda \geq 0$ and focus on the case where $0 \leq \lambda \leq 1$ in the last section.}
We begin by considering 
naive agents. These agents are unaware of their bias and as a result plan as if they will behave optimally in the future. A naive sunk cost bias agent, perceives the cost of stopping as $\lambda$ multiplied by the cost it invested. Therefore, %the expected payoff of a naive agent who already paid a total cost of $K$ is given by
%$$\pi_n(u,K) = \max\p{\sum_{v \in N(u)} p(u,v) \cdot \pi_o(v) -
%	c(u), -\lambda K }$$
%In particular, the naive agent 
it will continue traversing the graph at node $u$ after accumulating a sunk cost of $K$ if and only if $\sum_{v \in N(u)} p(u,v) \cdot \pi_o(v) -c(u) \geq -\lambda K$. For consistency, we use the term expected payoff, for all types of agents to denote their \emph{actual} expected payoff and not the perceived one. Thus, the expected payoff of the naive agent is $\pi_n(u,K) =0$ if it decides to abandon and is otherwise 
$$\pi_n(u,K) = \sum_{v \in N(u)} p(u,v) \cdot \pi_n(v,K+c(u)) -c(u).$$
%\subsection{Naive Sunk Cost Bias Agent}
It is not very hard to construct instances in which the payoff of the naive agent is negative. Here, we show that its negative payoff can grow exponentially in the graph's size. We first get the agent to incur sufficient 
sunk cost. Once the agent incurred this amount of sunk cost it would rather
continue with no chance of reaching any reward than cut its losses and stop.
Essentially, the agent prefers continuing a lost cause to admitting defeat. It
always believes that it will take just one more step and then give up, but once
it gets there, it finds that it no longer wants to give up.

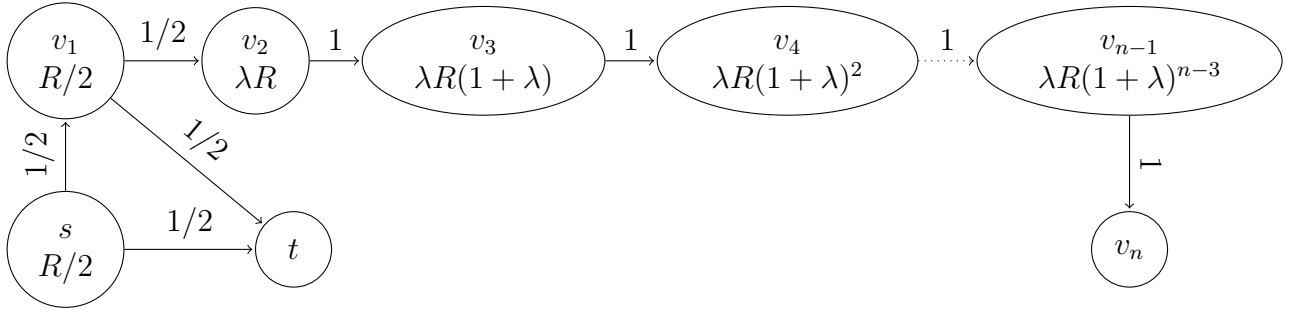
\begin{figure*}[t]
	\centering
	\begin{tikzpicture}[->,shorten >=1pt,auto,node distance=2cm, thin]
	\node (0) [circle, draw, inner sep=0pt,minimum size=1pt] at (0,0) 
	{\begin{tabular}{c} 
		$s$
		\\ $R/2$ 
		\end{tabular}};
	\node (1) [circle, draw, inner sep=0pt,minimum size=1pt] at (0,2.5) 
	{\begin{tabular}{c} 
		$v_1$
		\\ $R/2$ 
		\end{tabular}};;
	\node (2) [circle, draw,inner sep=0pt,minimum size=1pt] at (2.5,2.5) 
	{\begin{tabular}{c} 
		$v_2$
		\\ $\lambda R$ 
		\end{tabular}};
	\node (3) [ellipse, draw, inner sep=0pt,minimum size=1pt] at (5.5,2.5) 
	{\begin{tabular}{c} 
		$v_3$
		\\ $\lambda R(1+\lambda)$ 
	\end{tabular}};
	\node (4) [ellipse, draw, inner sep=0pt,minimum size=1pt] at (9.5,2.5) 
	{\begin{tabular}{c} 
		$v_4$
		\\ $\lambda R(1+\lambda)^2$ 
		\end{tabular}};;
	%\node (d) at (11.6,2.5) {$\cdots$};
	%\node(x) at (11.6, 2.8) {$1$};
	\node (e) [ellipse, draw, inner sep=0pt,minimum size=1pt] at (14, 2.5) 
	{\begin{tabular}{c} 
		$v_{n-1}$
		\\ $\lambda R(1+\lambda)^{n-3}$ 
		\end{tabular}};
	\node (vn) [ellipse, draw,  inner sep=0pt,minimum size=1cm] at (14,0) {$v_n$};
	\node (5) [ellipse, draw,  inner sep=0pt,minimum size=1cm] at (3,0) {$t$};
	
	\path[every node/.style={sloped,anchor=south,auto=false}]
	(0) edge node {$1/2$} (1)
	(1) edge node {$1/2$} (2)
	(2) edge node {$1$} (3)
	(3) edge node {$1$} (4)
	%(4) edge node {$1$} (d)
	(4) edge[dotted] node {$1$} (e)
	(e) edge node {$1$} (vn)
	
	(0) edge node {$1/2$} (5)
	(1) edge node  {$1/2$} (5)
	;
	\end{tikzpicture}
	\caption{Graph for which the expected payoff are exponentially negative.}
	\label{fig:exp_neg_expected}
\end{figure*}

\begin{claim}
	The negative expected payoff of a naive agent may be exponential in the graph's size.
\end{claim}
\begin{proof}
	Consider the instance depicted in Figure \ref{fig:exp_neg_expected}. Observe that a naive agent will choose to continue at $s$ since the expected payoff is $0$.\footnote{We assume that for naive agents ties are broken in favor of continuing.} If it ends up at $v_1$ it will again choose to continue, for the same reason, for any value of $\lambda$. Now for any $2\leq i \leq n-1$, the agent at $v_i$ will continue to $v_{i+1}$ since it will accumulate a cost of
%	\full{
	%\begin{align*}
	%	R+\sum_{j=0}^{i-3} \lambda R(1+\lambda)^{j}  &= R\left(1+\lambda R \cdot \frac{(1+\lambda)^{i-2}-1}{\lambda}\right)
	%	\\
	%&=  R \cdot (1+\lambda)^{i-2}
	%\end{align*}
%}
%	\aaai{
	\begin{align*}
	%{\textstyle
		R+\sum_{j=0}^{i-3} \lambda R(1+\lambda)^{j} =  R \cdot (1+\lambda)^{i-2}
	%}
	\end{align*}
%}
	 Thus, it will have a perceived cost of stopping of $ \lambda R \cdot (1+\lambda)^{i-2}$ which is the same as the cost it thinks it will have for continuing one step and then stopping. Since the expected payoff of the two first steps is $0$, we get that the expected payoff of the naive agent is 
	 %$R-R \cdot (1+\lambda)^{n-2}$.
	 ${-\frac{1}{4}{R}(1+\lambda)^{n-2}}$.
\end{proof}

\section{Sophisticated Sunk Cost Bias Agent} \label{sec-soph} %\nscomment{I've tried to add some explanations according to the reviews.} \nsedit{Even though a sophisticated is aware of its bias, it cannot mitigate it. In contrast to naive agent, a sophisticated sunk cost bias takes future actions to minimize the negative implications of its bias. The sophisticated agent takes deliberate actions in the present to avoid subjecting itself to temptation in the future. Corresponding to our model, the sophisticated agent plans to traverse the graph if and only if its future self would not end up with negative (actual) expected utility due to his bias. Formally, a} 

{Recall that a sophisticated agent is aware of its bias and hence it can compute its \emph{actual} expected payoff for continuing traversing the graph. Hence,} a sophisticated agent at node $u$ that accumulated a sunk cost of $K$ decides to continue if its (actual) expected payoff is {at least} $-\lambda K$. For this reason, a sophisticated agent will only traverse the graph if its expected payoff is non-negative. %and hence its expected payoff is always non-negative. 
Nevertheless, for any
$W$, we can construct an instance in which a sophisticated agent has an expected payoff of $0$ while the optimal agent has an expected payoff of $W>0$. Consider Figure~\ref{fig:scb_soph}, with $R = 10W$ and $\lambda = 0.5$. An optimal agent would continue at $s$ and stop at $u$, with expected payoff
$W$. However, a sophisticated agent would continue at $u$ as well, so its
expected payoff would be $0$. 
\begin{figure}[t]
	\centering
	\begin{tikzpicture}[->,shorten >=1pt,auto,node distance=2cm, thin]
	\node (s) [circle, draw, inner sep=0pt,minimum size=1.1cm] at (0,0) {\small $s;4W$};
	\node (u) [circle, draw, inner sep=0pt,minimum size=1.1cm] at (1.5,1.5) {\small $u;7W$};
	\node (v) [circle, draw, inner sep=0pt, minimum size=1.1cm] at (3,0) {\small $v$};
	\node (t) [circle, draw, inner sep=0pt, minimum size=1.1cm] at (1.5,-1.5) {\small $t$};
	
	\path[every node/.style={sloped,anchor=south,auto=false}]
	(s) edge node {\small $1/2$} (u)
	(s) edge node [below]  {\small $1/2$} (t)
	(u) edge node {\small $1/2$} (v)
	(u) edge node {\small $1/2$} (t)
	;
	\end{tikzpicture}
	\caption{Sophisticated agent gets 0, optimal agent gets $W$.}
	\label{fig:scb_soph}
\end{figure}
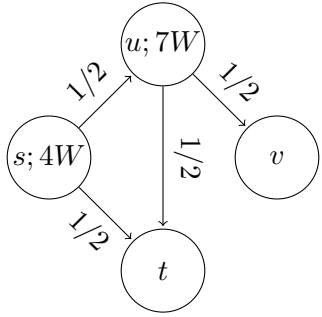

As we discussed, computing a policy for optimal or naive agents is computationally easy. We observe that this is not the case for sophisticated agents. Since a sophisticated agent is aware of its bias, in order to compute its expected payoff 
%A sophisticated agent faces a very complicated decision problem. In order to decide whether or not to traverse the graph 
it should know which action its future selves will take in each of the following nodes. This problem is complicated by the fact that the choice depends on the amount of sunk cost they accumulated in reaching these nodes. 
We observe that there are instances in which computing the expected payoff is computationally hard. Hence, it is computationally hard for a sophisticated agent to compute its policy. 

\begin{proposition} \label{prop-hardness-sophisticated}
	The problem of computing the optimal policy for a sophisticated sunk cost bias
	agent is \#P-hard.
\end{proposition}
\begin{proof}
	We use the following reduction
	from knapsack solution counting. We are given a capacity for the knapsack $B$ and $n$ objects with weights $w_1,\ldots w_n$. Our goal is to count the number of subsets $S$ such that $\sum_{i \in S} w_i \leq B$. Given an instance of the counting problem we construct the Markov decision process illustrated in Figure \ref{fig:knapsack_red}.
	
%	 following instance for computing the expected utility  for a sophisticated sunk cost biased
%	agent.  
	\begin{figure*}[t]
		\centering
		\begin{tikzpicture}[->,shorten >=1pt,auto,node distance=2cm, thin]
		\node (s) [circle, draw, inner sep=0pt,minimum size=1pt] at (0,0) {\begin{tabular}{c} $s$
			\\ {$C$} \end{tabular}};
		\node (u1) [circle, draw, inner sep=0pt,minimum size=1pt] at (2,1.5) {\begin{tabular}{c}
			$u_1$ \\ $w_1$ \end{tabular}};
		\node (v1) [circle, draw, inner sep=0pt,minimum size=1pt] at (2,-1.5) {\begin{tabular}{c}
			$v_1$ \\ $0$ \end{tabular}};
		\node (u2) [circle, draw, inner sep=0pt,minimum size=1pt] at (5,1.5) {\begin{tabular}{c}
			$u_2$ \\ $w_2$ \end{tabular}};
		\node (v2) [circle, draw, inner sep=0pt,minimum size=1pt] at (5,-1.5) {\begin{tabular}{c}
			$v_2$ \\ $0$ \end{tabular}};
		\node (u3) [circle, draw, inner sep=0pt,minimum size=1pt] at (8,1.5) {\begin{tabular}{c}
			$u_n$ \\ $w_n$ \end{tabular}};
		\node (v3) [circle, draw, inner sep=0pt,minimum size=1pt] at (8,-1.5) {\begin{tabular}{c}
			$v_n$ \\ $0$ \end{tabular}};
		\node (q) [circle, draw, inner sep=0pt,minimum size=1pt] at (10,0) {\begin{tabular}{c} $q$
			\\ $0$ \end{tabular}};
		\node (q') [ellipse, draw, inner sep=0pt,minimum size=1pt] at (14,0) {\begin{tabular}{c}
			$q'$ \\ $R + \lambda {(B+C)}$ \end{tabular}};
		\node (t) [circle, draw, inner sep=0pt,minimum size=0.8cm] at (12,-1.5) {$t$};
		
			\path[every node/.style={sloped,anchor=south,auto=false}]
		(s) edge node {\small 1/2} (u1)
		(s) edge node {\small 1/2} (v1)
		
		(u1) edge node {\small 1/2} (u2)
		(u1) edge node [above] {\small \hspace{25pt}1/2} (v2)
		(v1) edge node [above] {\small \hspace{-30pt} 1/2} (u2)
		(v1) edge node  {\small 1/2} (v2)
		
		(u2) edge node {\ldots} (u3)
		%	(u2) edge node [right] {0.5} (5,2)
		%	(v2) edge node [left] {0.5} (u3)
		(v2) edge node {\ldots} (v3)
		
		(u3) edge node {\small 1} (q)
		(v3) edge node {\small 1} (q)
		
		(q) edge node {\small 1/2} (q')
		(q) edge node {\small 1/2} (t)
		
		(q') edge node {\small 1} (t)
		;
		\end{tikzpicture}
		\caption{Reduction from knapsack solution counting to computing the payoff of a sophisticated sunk cost
			biased agent.}
		\label{fig:knapsack_red}
	\end{figure*}
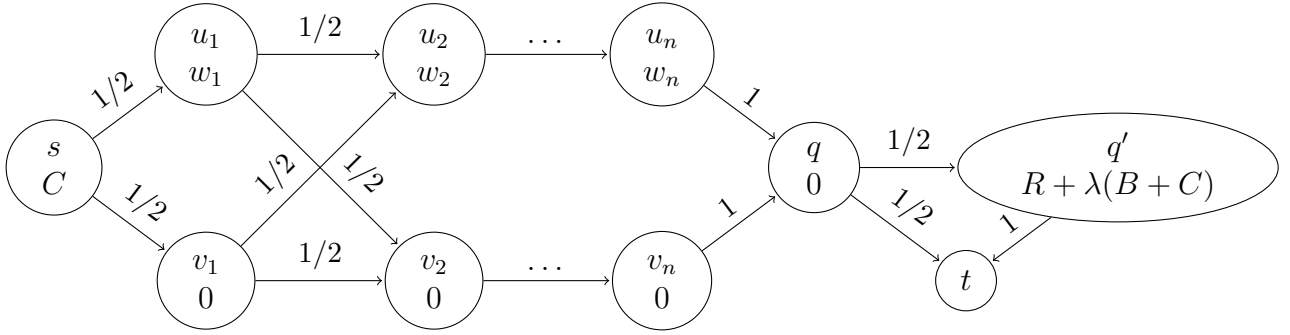
	
	%In Figure~\ref{fig:knapsack_red}, we argue that counting the number of subsets
	%of $\{w_1 \dots w_n\}$ such that the sum of the $w_i$'s is under some budget $B$
	%is equivalent to finding the expected payoff for the agent. 
	
	Let $K$ be the cost incurred after $s$
	and just before $q$, meaning $K$ is the sum of some subset of the $w_i$'s.
	We observe
	that at $q'$, the agent will have a different policy based on the value of $K$. Continuing at $q'$ yields a net payoff of $-\lambda (B + C)$, while the perceived cost of stopping is $-\lambda (K + C)$. Hence, the agent continues at $q'$ if and only if $K > B$. \footnote{For ease of presentation we break ties in favor of not continuing. This can be avoided by introducing small perturbations.}
%which is more than the
%	perceived cost of stopping, which is $-\lambda (K + C)$.
%	
%	If $K \le
%	B$, then the agent will stop at $q'$. However, if $K > B$, then it will continue
%	because this will yield a net payoff of $-\lambda (B + C)$, which is more than the
%	perceived cost of stopping, which is $-\lambda (K + C)$. 	
	Thus, we can separate all
	paths leading to $q$ into two types based on whether the agent would stop or continue at $q'$ if it reaches it. We first prove correctness under the assumption that the agent will always reach $q$ and then show it holds. 
	
	For paths along which the agent would not continue at $q'$, the expected payoff
	is $\pi_s^{\oslash} = \frac{R}{2}  - \E{K|K \le B} - C$. This is because such paths
	reach $t$ when a $q\rightarrow t$ transition is realized (with
	 probability $\frac{1}{2}$). Note that we can complete the reduction without knowing the value of $\E{K|K \le B}$. For paths along which the agent would continue at $q'$ the expected payoff is
{	 \begin{align*}
	\pi_s^{\copyright} = \frac{R}{2} - \E{K|K > B} - \frac{\lambda}{2}(B +C) -C.
	\end{align*}}
This is because these paths always reach $t$, but they incur a cost of $R +
	\lambda (B + C)$ with probability $\frac{1}{2}$. Let $p$ be the probability that $K \le B$. The total expected payoff is then
{	
	\begin{align*}
	\pi_s &= p \cdot \pi_s^{\oslash} + (1-p) \cdot \pi_s^{\copyright} \\
%	&= p \p{\frac{R}{2} - \E{K|K \le B} - C} \\
%	&+ (1-p) \p{\frac{R}{2}  -
%		\E{K|K > B} - \frac{\lambda}{2} B - C(1 + \frac{\lambda}{2})} \\
&= p  \frac{R}{2}  + (1-p)\frac{R}{2} - p\cdot \E{K|K \le B} - (1-p)\cdot \E{K|K > B} \\
&-p\cdot C -(1-p)\cdot \left(\frac{\lambda}{2}(B +C) + C \right) \\
	&= \frac{R}{2} - \frac{1}{2} \sum_{i=1}^n w_i -C ~-\frac{\lambda}{2}(1-p)(B+C)
	\end{align*}
}
%\normalsize
	where in the last step we use:
	%{
		$$ p\cdot \E{K|K \le B} + (1-p)\cdot \E{K|K > B} = \E{K} = \frac{1}{2}\sum_{i = 1}^{n}w_i.$$
	By rearranging we get that:
	%SO-19/2 - removed step
%	{
%	\begin{align*}
%%	\pi_s &= \frac{R}{2}  - \frac{\sum_{i=1}^n w_i}{2} - \frac{\lambda}{2} (1-p) B \\
%	\pi_s  + \frac{\sum_{i=1}^n w_i}{2} - \frac{R}{2} + C = (p-1) \frac{\lambda}{2}(B + C)
%		\end{align*}
%	}
%		This implies that
%%	{	
\begin{align*}
	 p = \frac{2\pi_s + \sum_{i=1}^n w_i - R + 2C}{\lambda (B + C)} + 1
	\end{align*}
%}
	Since the number of paths such that $K \le B$ is $p \cdot 2^n$, and the number of
	such paths is the same as the number of solutions to the knapsack problem, we
	have
	$${
	\text{\# Solutions} = 2^n \p{	\frac{2\pi_s + \sum_{i=1}^n w_i - R + 2C}{\lambda (B + C)} + 1}.}$$
	\noindent Hence, if the sophisticated agent can compute its expected payoff in polynomial time, it can also solve the \#P-hard problem of knapsack solution counting. To complete the reduction, we first show that the agent will always
	reach $q$ if it starts to traverse the graph. To do this, we observe that from any node before $q$ (not including $s$) the expected payoff is at least $\frac{R}{2} -
	\lambda (B + C) - \sum_{i=1}^n w_i$. Therefore, the agent will continue as long as
	\begin{align*}
	{
	\frac{R}{2} -
	\lambda (B + C) - \sum_{i=1}^n w_i {\geq} -\lambda C.}
	\end{align*}
	By rearranging we get the reduction holds for any $R {\geq} 2\sum_{i=1}^n w_i + 2\lambda B $.
	
%	 that as long as we choose $R {\geq} 2\sum_{i=1}^n w_i + 2\lambda B $, the above reduction holds.

	Finally, recall that we need to show that the problem of computing the optimal policy for a sophisticated agents is \#P-hard. This requires us to construct an instance in which the agent has to know the exact expected payoff to compute its policy. We prove the following claim:
	\begin{claim} \label{clm:binary-search}
	For any $0<\alpha<1$ there exists $C$ and $R$ such that a sophisticated agent will traverse the graph if and only if 
	$\# \text{Solutions} \geq \alpha \cdot 2^n$.
	\end{claim} 
\begin{proof}
Let $R = 2\sum_{i=1}^n w_i + 2\lambda B$. We show that for any $0<\alpha<1$ there exists $C$ such that the expected payoff of the sophisticated agent is $0$ if $p=\alpha$. This implies that the agent will traverse the graph if and only if $p\geq \alpha$ as required. Observe that
%Given the previous discussion we need to show that there exists $C$ such that the following holds:
\begin{align*}
\pi_s &= {\frac{R}{2} - \frac{1}{2}\sum_{i = 1}^{n} w_i -C -\frac{\lambda}{2}(1-\alpha)(B+C) = 0} \\
\implies C &={\frac{R - \sum_{i = 1}^{n} w_i -\lambda(1-\alpha)B}{2 + \lambda(1 - \alpha)}} 
\end{align*}
 % +\eps$, where $\eps>0$ is some small constant. 
%\nscomment{This yields $C = \frac{\sum_{i = 1}^{n}w_i +\lambda B(1 + \alpha) + \varepsilon}{2 + \lambda(1-\alpha)}$ which is positive and defined}.
Finally, notice that both $R$ and $C$ have a polynomial representation
as they are defined by arithmetic operations over numbers that their representation is polynomial in the problem's size.
\end{proof}

	The claim implies the proof of Proposition \ref{prop-hardness-sophisticated} since if the sophisticated agent can compute its optimal policy for any instance constructed for $0<\alpha<1$, we can use its poly-time algorithm to run a binary search over the fractions of valid solutions. Since the size of the search space is $2^n$, our binary search will compute the exact number of solutions to the knapsack problem in polynomial time.
\end{proof}

Note that even though a naive agent can efficiently decide on its actions (as it plans to behave optimally), the same reduction as in the proof of Proposition \ref{prop-hardness-sophisticated} shows that computing the expected payoff of a naive agent is also \#P-Hard.

\subsection{Bounding the Payoff of  a Sophisticated Agent}
In this section, we try to further understand how much smaller the payoff of a sophisticated agent may be relative to the payoff of an optimal agent. As discussed in the example in Figure \ref{fig:scb_soph}, we cannot have a multiplicative bound here as in some cases the payoff of the optimal agent can be positive while the payoff of a sophisticated agent is $0$. {As the payoff of an optimal agent is at most $R$ we have that for any value of $\lambda\geq 0$, $\pi_s\geq \pi_o-R$. In fact, this bound is asymptotically tight even for a $3$-node graph as $\lambda$ is approaching infinity. We prove this by showing in Appendix \ref{app-3-nodes} that for any value of $\lambda\geq0$ there exists a 3-node instance such that, $\pi_s = \pi_o - \frac{2+\lambda -2\sqrt{1+\lambda}}{\lambda} \cdot R$. This implies that there exists a family of instances such that as $\lambda$ is approaching infinity $\pi_o-\pi_s$ is approaching $R$.
While the trivial bound $\pi_s\geq \pi_o-R$ is asymptotically tight for very large values of $\lambda$, as we will see next, it is far from tight for smaller values of $\lambda$. Hence, for the rest of the paper, we look for tighter bounds on the payoff of sophisticated agents.
}

To help analyze the payoff of the sophisticated agent we define an auxiliary ``hybrid'' agent with the following behavior: it continues wherever
the optimal agent continues until it either reaches the goal or the optimal
agent stops. If the optimal agent stops, then it does what the sophisticated
agent would have done. We show that the payoff of the sophisticated agent is always higher than the expected payoff of the ``hybrid'' agent and then provide bounds on the expected payoff of a hybrid agent.

\begin{proposition} \label{prop-hyb}
The expected payoff of a sophisticated agent is at least the expected payoff of a hybrid agent.
\end{proposition}
\begin{proof}
We label the nodes of the graph according to the manner in which the hybrid agents behaves: nodes that it behaves as an optimal agent are labeled by $o$ and those at which it behaves as a sophisticated agent by $s$. In particular, nodes that the optimal agent reaches and decides to continue are labeled by $o$ and the rest of the nodes are labeled by $s$. Since the decisions of the optimal agent do not depend on the path leading to a node this is a well defined distinction. Moreover, once the hybrid agent starts behaving as a sophisticated agent it would keep doing so until it either reaches the target or stops. Thus, starting from any $s$-labeled node the sophisticated and hybrid agents will behave the same.

{
Denote the payoff of the hybrid agent by $\pi_h$. We show that for every $o$-labeled node $u$ and every $K$ which is the weight of a path from $s$ to $u$ we have that $\pi_h(u,K) \leq \pi_s(u,K)$. Assume towards a contradiction that there exists an $o$-labeled node $u$ and a path reaching this node of cost $K$ such that $\pi_h(u,K) > \pi_s(u,K)$. If there is more than one such node, let $u$ be the last such node in the topological order. This implies that for every node $v$ subsequent to it and any $K'>0$ a cost of a path reaching $v$ we have that $\pi_h(v,K') \leq \pi_s(v,K')$.}

%Such a node exists as for any $s$-labeled node $z$ and for any $K'>0$ we have that $\pi_h(z,K') = \pi_s(z,K')$. 
Since in all $o$-labeled nodes the optimal agent decides to continue there are two cases. In the first case, at $u$ the sophisticated agent decides to continue as well, thus the expected payoff is the weighted average over all
successor nodes, and hence $\pi_h(u,K) \leq \pi_s(u,K)$.  In
the second case, the sophisticated agent stops and the hybrid agent continues.
We know that the expected payoff of the sophisticated agent if
it continued would be non-positive, as the sophisticated agent decides to
stop, and the expected payoff of the hybrid agent is the weighted average over
all successors. Again, by our assumption on subsequent nodes, we have that
% that for any subsequent node $v$ and any $K'>0$ we have that $\pi_h(v,K') \leq \pi_s(v,K')$,
the weighted average over
these successors for the hybrid agent must be no more than the same weighted
average for the sophisticated agent, which we know is no more than 0. Thus,
$\pi_h(u,K)\leq \pi_s(u,K)$.
\end{proof}

In order to lower-bound the payoff of a sophisticated agent, we now take a closer look at the payoff of an optimal agent.  Let $S$ denote the event in which the optimal agent reaches $t$ (starting from $s$) and $p(S)$ the probability of this event. The payoff of an optimal agent is: $\pi_o = p(S)\cdot R -E[C]$, where $C$ is a random variable equal to the cost that the optimal agent incurred. By decoupling the event in which the optimal agent reaches the target and does not reach the target:
\begin{align*}
{\pi_o = p(S)\cdot R -p(S)\cdot E[C|S] - (1-p(S))\cdot E[C|\bar S]}
\end{align*}
Based on this we can bound the payoff of the hybrid agent:
\begin{lemma}
	 $\pi_h \geq \pi_o - \lambda\cdot (1-p(S))\cdot E[C|\bar S]$
\end{lemma}
\begin{proof} The hybrid agent reaches the target with probability at least $p(S)$ since whenever the optimal agent reaches the target the hybrid agent reaches it as well. In this case its expected cost would be $E[C|S]$. Wherever the optimal agent 
	stops, the hybrid agent follows the sophisticated agent's actions. Assume that
	the hybrid agent has incurred some cost $K$ in reaching a node $u$ at which
	the optimal agent stops. At this point, using the decision rule for the
	sophisticated agent, we know that $\pi_s(u) \ge -\lambda K$. In particular, this implies 
	that the extra expected cost for continuing is at most $\lambda K$. Thus, the hybrid
	agent can incur a total cost of at most $(1+\lambda)K$ along this path. 
	Taking this as an expectation over all paths in which the optimal agent stopped prematurely, we find that the expected cost incurred
	by the hybrid agent is at most $(1+\lambda)E[C|\bar S]$. Thus, the expected payoff of
	the hybrid agent is at least:
	\begin{align*}
	{
	p(S) \cdot R -p(S)\cdot E[C|S] -(1-p(S))(1+\lambda )E[C|\bar S] 
}
%	\\ &=\pi_o - \lambda\cdot (1-p(S))\cdot E[C|\bar S]
	\end{align*} 
	which equals $\pi_o - \lambda\cdot (1-p(S))\cdot E[C|\bar S].$
\end{proof}

By applying Proposition \ref{prop-hyb} we conclude that
\begin{corollary} \label{cor-hyb}
$\pi_s \geq \pi_o - \lambda\cdot (1-p(S))\cdot E[C|\bar S]$.
\end{corollary}
Next, we derive two more specific bounds: %by bounding $(1-p(S))\cdot E[C|\bar S]$.
\begin{claim}
$\pi_s \geq \pi_o - \lambda\cdot p(S)\cdot R$.
\end{claim}
\begin{proof}
As we know that optimal agent always has a {non-negative} payoff we have that:
\begin{align*}
p(S)&\cdot R -p(S)\cdot E[C|S] - (1-p(S))\cdot E[C|\bar S] \geq 0 \\
&\implies (1-p(S))\cdot E[C|\bar S] \leq p(S)\cdot R -p(S)\cdot E[C|S] \\
&\implies (1-p(S))\cdot E[C|\bar S] \leq p(S)\cdot R 
\end{align*}

By applying Corollary \ref{cor-hyb} we have that
%\begin{align*}
${\pi_s \geq \pi_o - \lambda\cdot (1-p(S))\cdot E[C|\bar S] \geq \pi_o - \lambda\cdot p(S)\cdot R}$ 
%\end{align*}
\end{proof}

We can get a closed form bound by using the fact that $\pi_s\geq 0$:
\begin{proposition} \label{prop:not-tight-bound-soph}
	$\pi_s \geq \pi_o - \frac{\lambda}{1+\lambda} \cdot R$.
\end{proposition}
\begin{proof}
	If $\pi_o \leq \frac{\lambda}{1+\lambda} \cdot R$, then the claim holds simply because $\pi_s\geq 0$. Else, $\pi_o \geq \frac{\lambda}{1+\lambda} \cdot R$. Thus we have that:
	\begin{align*}
	\pi_o &= p(S)\cdot R -p(S)\cdot E[C|S] - (1-p(S))\cdot E[C|\bar S] \geq \frac{\lambda}{1+\lambda} \cdot R 
	\end{align*}
	By rearranging we get that:
	\begin{align*}
(1-p(S))\cdot E[C|\bar S] &\leq p(S)(R- E[C|S]) - \frac{\lambda}{1+\lambda} \cdot R \\
 &\leq R - \frac{\lambda}{1+\lambda} \cdot R \leq  \frac{1}{1+\lambda} \cdot R
	\end{align*}
This implies that $\pi_s \geq \pi_o - \frac{\lambda}{1+\lambda} \cdot R$ as required.
\end{proof}

In 
Appendix \ref{app-3-nodes} 
%the supplementary material
we provide a complete analysis of $3$-node graphs and show that %$\pi_s \geq \pi_o - \frac{(2+\lambda -2\sqrt{1+\lambda})}{\lambda} \cdot R$
for such graphs this bound is not tight. {This is done by showing that
for any $3$-node graph $\pi_s \geq \pi_o - \frac{2+\lambda -2\sqrt{1+\lambda}}{\lambda} \cdot R$ and noting that for any $\lambda\geq 0$ we have that $\frac{\lambda}{1+\lambda} > \frac{2+\lambda -2\sqrt{1+\lambda}}{\lambda}$. The gap between the two bounds increases as $\lambda$ approaches $0$, hence, in the next section we focus on the case of $0\leq \lambda \leq 1$ and present a tighter bound for this specific and natural topology. We present some evidence that this bound is, in a sense, asymptotically tight. It is interesting to note that the bound $\pi_s \geq \pi_o - \frac{\lambda}{1+\lambda} \cdot R$ is tight for any value of $\lambda\geq 0$ in an extended model} in which the cost for continuing at a certain node may depend on the edge that was taken. In other words, the costs are on the edges instead of on the nodes. We provide this proof in %the supplementary material 
Appendix \ref{app-3-nodes} as well.
\section{Analysis of the Fan Graph} \label{sec-fan} 
In this section we focus on $0 \leq \lambda \leq 1$, motivated by our goal of finding tighter bounds when $\lambda$ is bounded. We consider a specific graph topology which is relatively easy to analyze: the fan. A fan graph consists of a path $(s=v_1,v_2,\ldots,v_n,t)$ such that for any $1\leq i<n$: $p(v_i,v_{i+1}) = p_i,~p(v_i,t) = 1-p_i$ and $c(v_i)=c_i$. Also, $p(v_n,t)=1$ and $c(v_n)=c_n$. A sketch of a \emph{specific} fan is depicted in  Figure \ref{fig:fan-tight}.
{Fan graphs can be used to capture, for example, scenarios of project development: at every step some cost should be invested to continue the project and then with some probability it is successful. }

For fans we derive an essentially asymptotically tight bound on the difference between the payoff of a sophisticated agent and an optimal agent. We present the bound and then demonstrate its tightness:
\begin{theorem} \label{thm-fan-bound}
	In every fan graph $\pi_s \geq \pi_o - \lambda \cdot(1-\frac{1}{n})^{n} \cdot R$.
\end{theorem}
\begin{proof}
	Recall that by Corollary \ref{cor-hyb} we have that $\pi_s \geq \pi_o - \lambda\cdot (1-p(S))\cdot E[C|\bar S]$, where $S$ is the event that the optimal agent reached the target and $ E[C|\bar S]$ is the expected cost of the optimal agent for paths in which it stopped before the target. To bound the payoff of the sophisticated agent we bound  $(1-p(S))\cdot E[C|\bar S]$. Assume without loss of generality that the optimal agent stopped traversing at $v_k$. 
In the fan there is only a single failing path $(s,v_2,\ldots,v_k)$. Thus, 
	\begin{align} \label{eq-fan}
		(1-p(S))\cdot E[C|\bar S] = (\prod_{i=1}^{k-1} (1-p_i)) \cdot (\sum_{i=1}^{k-1} c_i)
	\end{align}
	 We now use the fact that the optimal agent traverses the graph till $v_{k}$ to get an upper bound on $(1-p(S))\cdot E[C|\bar S]$. %In particular we claim that since the optimal agent 
	\begin{lemma}  \label{lem:fan-upper-expected}
		If the optimal agent reaches $v_k$ then $\sum_{i=1}^{k-1} c_i \leq \sum_{i=1}^{k-1} p_i \cdot R$.
	\end{lemma}
	\begin{proof}
		 Let $\pi_o(v_j,v_k)$ denote the payoff of an {optimal} agent traversing the graph from node $v_j$ to $v_k$. We prove by a backward induction that for any $j\geq 1$, $\pi_o(v_j,v_k) \leq \sum_{i=j}^{k-1} \left( p_i\cdot R -c_i \right)$.  For the base case observe that if the agent continues from $v_{k-1}$ to $v_k$ it is indeed the case that $\pi_o(v_{k-1},v_k) =  p_{k-1}\cdot R - c_{k-1}$. For the induction step we assume correctness for $v_{j+1}$ and prove it for an agent traversing the graph starting from $v_j$. Observe that, if the agent traverses the graph from $v_j$ to $v_k$ then: $\pi_o(v_j,v_k) = p_j\cdot R -c_j + (1-p_j) \cdot \pi_o(v_{j+1},v_k)$. By the induction hypothesis we have that $\pi_o(v_{j+1},v_k) \leq \sum_{i=j+1}^{k-1} \left( p_i \cdot R -c_i \right)$ putting this together we get that $\pi_o(v_j,v_k) \leq \sum_{i=j}^{k-1} \left( p_i \cdot R -c_i \right)$ as required. Thus, the expected payoff of the optimal agent reaching $v_k$ is at most $\sum_{i=1}^{k-1} \left( p_i \cdot R -c_i \right)$. 
Since this quantity has to be non-negative, we have $\sum_{i=1}^{k-1} c_i \leq \sum_{i=1}^{k-1} p_i \cdot R$.
	\end{proof}
%	\begin{corollary}

By applying Lemma \ref{lem:fan-upper-expected} on Equation \ref{eq-fan} we get that
%Thus, we have that if the optimal agent stops at $k$, then $\sum_{i=1}^{n-1} c_i \leq \sum_{i=1}^{n-1} p_i \cdot R$. Putting this together we get that
	\begin{align*}
	(1-p(S))\cdot E[C|\bar S] \leq
	(\prod_{i=1}^{k-1} (1-p_i)) \cdot (\sum_{i=1}^{k-1} R \cdot p_i) 
	\end{align*}
	In Claim \ref{clm:optimizing} in  
	%the supplementary material 
	the appendix we show that the maximum of this function is attained when for every $1\leq i<k$: $p_i = \frac{1}{k}$ implying
	\begin{align*}
	(1-p(S))\cdot E[C|\bar S] \leq R \cdot  (k-1)\cdot \frac{1}{k} \cdot (1-\frac{1}{k})^{k-1} 
%	\\
	&
	= R\cdot (1-\frac{1}{k})^{k}
	\end{align*}
		thus, the maximum is attained for $k=n$.
\end{proof}

By noticing that when $n$ approaches infinity $(1-\frac{1}{n})^{n}$ approaches $\frac{1}{e}$ we conclude that:
\begin{corollary} \label{cor-fan-bound}
	In every fan graph $\pi_s \geq \pi_o -\frac{1}{e} \cdot \lambda \cdot R$.
\end{corollary}

%\begin{conjecture}
%	For any fan of size $n$ and any $\lambda$ there exists a fan graph maximizing the gap between $\pi_s$ and $\pi_o$ such that for every $i$: $p_i=\frac{1}{n}+\eps$ and $c_i=\frac{1}{n}+\eps$.
%\end{conjecture}
\begin{figure}[t]
	\centering
	\includegraphics[width=1\linewidth]{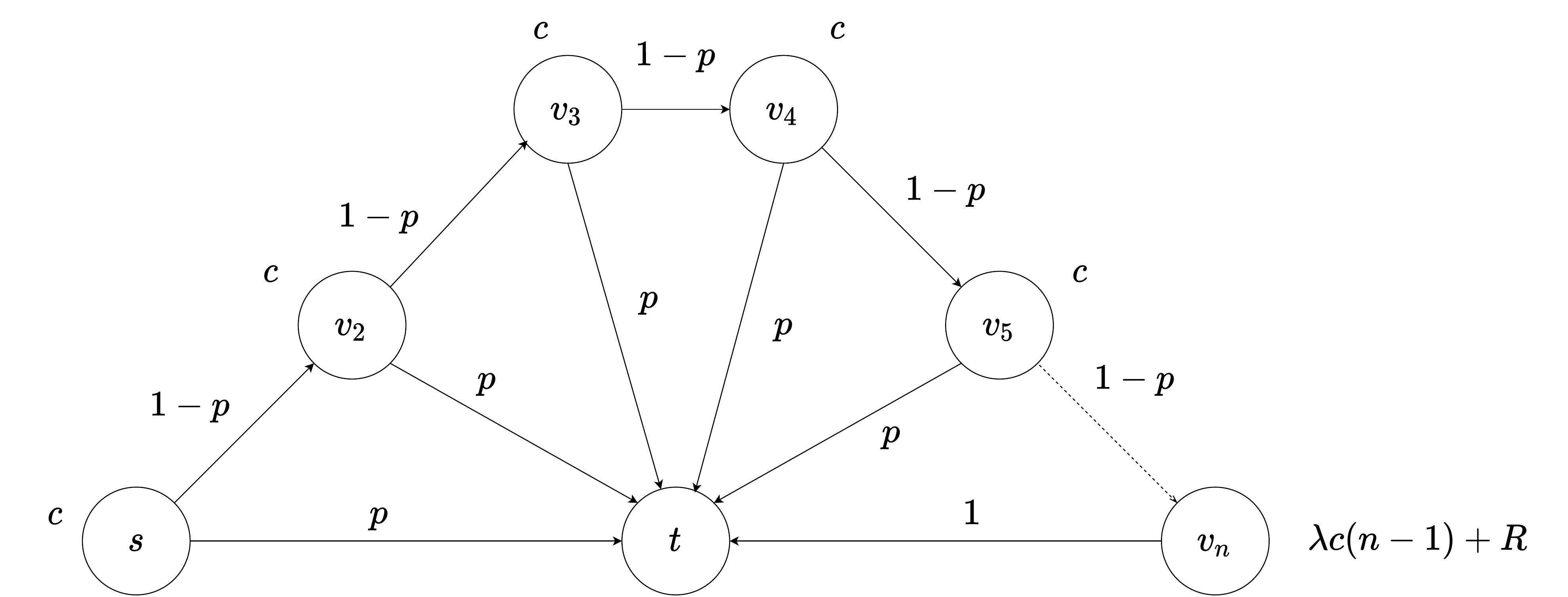}
	\caption{An instance of the Fan graph for Theorem \ref{thm-fan-bound}.}
	\label{fig:fan-tight}
\end{figure}
Next, we show this bound is, in some sense, tight in the limit:
\begin{theorem} \label{thm-fan-tight-bound}
	There exists a family of instances in which $\lambda$ is a function of $n$, such that as $n$ goes to infinity $\lim_{n \rightarrow \infty} \pi_o -\frac{1}{e} \cdot \lambda \cdot R = \pi_s$.
\end{theorem}
\begin{proof}
We consider fan graphs in which the transition probabilities and costs on the path are identical. That is, for every $1\leq i < n$, $p_i=p = \frac{1}{n}$ and $c_i = c = \frac{1}{n} - \frac{1}{n^2}$. We also set $R=1$. The graph is illustrated in Figure \ref{fig:fan-tight}. 

It is easy to see that the optimal agent would traverse the graph till it reaches $v_n$ as the expected payoff of each step by itself is strictly positive. The optimal agent would stop at $v_n$ as the expected payoff of this last step is negative. The expected payoff of the optimal agent is
{ \begin{align*}
{\pi_o = \sum_{i=1}^{n-1} (1-p)^{i-1} (p-c)   = \frac{1 - (1-p)^{n-1}}{p}(p-c)}
\end{align*} }
%and would then stop. Observe that the expected payoff of the optimal agent is 
%
% In addition, observe that the expected payoff of an optimal agent that starts traversing the graph from node $v_k$ is $\frac{1 - (1-p)^{n-k}}{p}(p-c)$, since at any step the expected payoff is $p + (1-p)\pi_o(v_{k+1}) - c$. Therefore, the expected payoff of an optimal agent that starts traversing the graph from $v_1$ is $\pi_o = \frac{1 - (1-p)^{n-1}}{p}(p-c)$. With $p > c$ its easy to see that the optimal agent will reach $v_n$.
%	
Next, we let $\lambda = \frac{\pi_o}{(1-p)^{n-1}(n-1)c}$. In Lemma \ref{lem-lambda-small-1} in the appendix, we show that it is indeed the case that $\lambda\leq 1$. %(Claim \ref{lambda-small-1}).

%SO 19/2 - removed the whole computation
Observe that $\lim_{n \rightarrow \infty} (1-p)^{n-1}(n-1)c = \frac{1}{e}$.
%\begin{align*}
%&{\lim_{n \rightarrow \infty} (1-p)^{n-1}(n-1)c}
%\\
%&{= \lim_{n \rightarrow \infty} (1 - \frac{1}{n})^{n - 1}(n-1)(\frac{1}{n} - \frac{1}{n^2})}
%\\
%&{= \lim_{n \rightarrow \infty} (1 - \frac{1}{n})^{n} - \frac{1}{n}(1 - \frac{1}{n})^n = \frac{1}{e}}.
%\end{align*}
Thus, we have that as $n$ approaches infinity $\pi_o$ is approaching $ \frac{\lambda}{e}$. To complete the proof we show that for this value of 
%
% and show that for this value of $\lambda$: (1) the payoff of the sophisticated agent is $0$ and (2) $\pi_o$ approaches $\lambda / e$. We begin with the second statement:
%\begin{align*}
%{\textstyle\lim_{n \rightarrow \infty} \frac{\pi_o}{\lambda}} &{\textstyle= \lim_{n \rightarrow \infty} (1-p)^{n-1}(n-1)c}
%\\
%&{\textstyle= \lim_{n \rightarrow \infty} (1 - \frac{1}{n})^{n - 1}(n-1)(\frac{1}{n} - \frac{1}{n^2})}
%\\
%&{\textstyle= \lim_{n \rightarrow \infty} (1 - \frac{1}{n})^{n} - \frac{1}{n}(1 - \frac{1}{n})^n = \frac{1}{e}}
%\end{align*}
%Therefore, as $n$ goes to infinity, $\pi_o$ goes to $\lambda \cdot \frac{1}{e}$. 
%We prove that for the chosen value of 
$\lambda$, $\pi_s=0$. To this end we show that if the agent starts traversing the graph it will go all the way to $t$ and in this case its expected payoff will be $0$. Observe that the expected payoff of a sophisticated agent that continues from node $v_k$ till $t$ is:
	\begin{align*}
\pi_s(v_k) = \pi_o(v_k) - \lambda (1-p)^{n-k}(n-1)c
	\end{align*}
	The first term is the payoff of the sophisticated agent for getting to $t$ before $v_n$ which is just the same as the optimal agent and the second term is the cost that the sophisticated agent incurs for taking the last step from $v_n$ to $t$. Thus, the expected payoff of a sophisticated agent that starts at $s$ and only stops when reaching $t$ is %\aaai{$\pi_s(v_1) = \pi_o - \lambda (1-p)^{n-1}(n-1)c$}
%\full{	
	 \begin{align*}
	 \pi_s(v_1) = \pi_o - \lambda (1-p)^{n-1}(n-1)c
	 \end{align*}
%}
	 which is $0$ for our choice of $\lambda$.
	
	We now prove that a sophisticated agent that starts traversing the graph will go all the way to $t$. We give a proof by a backward induction. That is, we assume that if the agent arrives at $v_{k+1}$ from $v_k$ then it will go all the way to $t$. For the base case observe that if the agent arrives at $v_n$, then the payoff for abandoning is $-\lambda c(n-1)$. On the other hand, moving to $t$ incurs the same payoff: $R - (\lambda c(n-1) + R) = -\lambda c(n-1)$ which is the same as quitting. For ease of presentation we assume that the agent breaks ties in favor of continuing and hence the agent will choose to continue.\footnote{To avoid applying a tie-breaking rule we can add small perturbations to the costs to make sure that the agent will strictly prefer to continue.} For the induction step we assume correctness for $v_{k+1}$ and prove for an agent traversing the graph from $v_k$. By the induction hypothesis, since the sophisticated agent continues to traverse the graph from $v_{k+1}$ till $t$ then its expected payoff for continuing is $\pi_s(v_k)$. To complete the proof we show that the
	expected payoff is greater than or equal to the agent's sunk cost (i.e., $\pi_s(v_k) \geq -\lambda c(k-1)$):
	
	\begin{lemma} \label{lem-fan-tight}
		$\pi_s(v_k) \geq -\lambda c(k-1)$.
	\end{lemma}
	\begin{proof} 
		We need to show that:
		\begin{align*}
		\pi_o(v_k) -\lambda (1-p)^{n-k}(n-1)c \geq -\lambda c(k-1)
		\end{align*}
		By rearranging we get that:
		\begin{align*}
		\pi_o(v_k)  \geq \lambda c \left( (1-p)^{n-k}(n-1) -(k-1) \right) 
		\end{align*}
		Since $\lambda c > 0$, if $((1-p)^{n-k}(n-1) -(k-1)) \leq 0$ the lemma trivially holds. 
		Else, assume that $((1-p)^{n-k}(n-1) -(k-1)) > 0$, hence we can divide by this and get that
		\begin{align*}
		\lambda \leq \frac{\pi_o(v_k)}{c((1-p)^{n-k}(n-1) - (k-1))} 
		\end{align*}
		By substituting for our choice of $\lambda$ we get that:
		\begin{align*}
		\frac{\pi_o}{(1-p)^{n-1}(n-1)c} \leq \frac{\pi_o(v_k)}{c((1-p)^{n-k}(n-1) - (k-1))} 
		\end{align*}
		By rearranging we get that:
		\begin{align*}
		\frac{1 - (1-p)^{n-1}}{(1-p)^{n-1}(n-1)} \leq \frac{1 - (1-p)^{n-k}}{(1-p)^{n-k}(n-1) - (k-1)} 
		\end{align*}
		which implies that
		\begin{align*}
		\underbrace{(1-p)^{n-k}(n-1) - (k-1) + (1-p)^{n-1}(k-1)}_{f(k)} &\leq
		\\
		(1-p)^{n-1}(n-1)
		\end{align*}
		%Therefore we need to prove that:
		%\begin{align*}
		%(1-p)^{n-k}(n-1) - (k-1) + (1-p)^{n-1}(k-1) \leq (1-p)^{n-1}(n-1)	
		%\end{align*}
Finally, we apply Lemma \ref{func:convex} to show that the above inequality holds. This is done by proving that $f(k)$ is bounded from above by $(1-p)^{n-1}(n-1)$ for any integer $1 \leq k \leq n$. 
		%	The proof of this auxiliary claim can be found in Claim \ref{func:convex}.
		%	In claim \ref{func:convex} we show that the function $f(x) = (1-p)^{n-x}(n-1) - (x-1) + (1-p)^{n-1}(x-1)$ where $0 < p < 1$ is bounded above by $(1-p)^{n-1}(n-1)$ for each $1 \leq x \leq n$ which implies the claim.
	\end{proof}

	\begin{lemma} \label{func:convex} 
		The function $f(x) = (1-p)^{n-x}(n-1) - (x-1) + (1-p)^{n-1}(x-1)$ where $0 < p < 1$ is bounded above by $(1-p)^{n-1}(n-1)$ for each $1 \leq x \leq n$.
	\end{lemma}
	\begin{proof}
		In order to prove this claim we show that $f$ is convex in $[1, n]$ and that $f(1) = f(n) = (1-p)^{n-1}(n-1)$. Therefore, for each $x \in [1, n]$, $f(x) \leq (1-p)^{n-1}(n-1)$.
		\newline
		\newline
		Observe that indeed $f(1) = f(n) = (1-p)^{n-1}(n-1)$. In addition:
		\begin{align*}
		f^{'}(x) &= -\ln(1-p)(n-1)(1-p)^{n-x} + (1-p)^{n-1} - 1
		\\
		f^{''}(x) &= \ln^{2}(1-p)(n-1)(1-p)^{n-x}
		\end{align*}
		Thus, for $1 < x < n$ and $0 < p < 1$ we have that $f^{''}(x) \geq 0$. Therefore, $f(x) \leq (1-p)^{n-1}(n-1)$ for each $x \in [1, n]$ as required.
	\end{proof}

%	In Lemma \ref{lem-fan-tight} in the appendix we prove that this 
% (i.e., $\pi_s(v_k) \geq -\lambda c(k-1)$) and hence the agent will continue traversing the graph, which completes the proof.
	\end{proof}

\section{Conclusion and Further Directions}

Sunk cost bias is a key behavioral bias that people exhibit, and 
it interacts in complex ways with uncertainty about the future.
We have proposed a model that allows us to study this interaction,
by analyzing the loss in performance of
agents that experience sunk cost bias as they
perform a planning problem with stochastic transitions between states.

There are a number of further questions suggested by this work.
A central open question is whether fan graphs represent the worst case
for sophisticated agents with sunk cost bias.
It would also be interesting to explore how the stochastic environment
for sunk cost bias can be adapted to incorporate agents with other
kinds of biases as well.

%\begin{contributions} % will be removed in pdf for initial submission,
                      % so you can already fill it to test with the
                      % ‘accepted’ class option
    %Briefly list author contributions.
    %This is a nice way of making clear who did what and to give proper credit.

    %H.~Q.~Bovik conceived the idea and wrote the paper.
    %Coauthor One created the code.
    %Coauthor Two created the figures.
%\end{contributions}

%\begin{acknowledgements} 
%	Work supported in part by BSF grant 2018206, Vannevar Bush Faculty Fellowship, MURI grant W911NF-19-0217, AFOSR grant FA9550-19-1-0183 and ISF grant 2167/19.
%\end{acknowledgements}

\newpage
\bibliography{sunk}

\newpage
\onecolumn
\appendix
\section{Missing Proofs from Section 4} \label{app-aux}
\begin{claim} \label{clm:optimizing} 
	The function $\prod_{i=1}^{n-1} (1-p_i) \cdot (\sum_{i=1}^{n-1} p_i)$ attains its maximal value for $0\leq p_i \leq 1$, when for every $i$, $p_i = 1/n$.
\end{claim}
\begin{proof}
	We simply take partial derivatives and compare them to $0$. To this end, it will be more convenient to use $q_i = 1-p_i$ and take partial derivatives of the function 
	$$f(q_1,\ldots,q_n) = \prod_{i=1}^{n-1} q_i \cdot (n-1- \sum_{i=1}^{n-1} q_i).$$
	Observe that:
	\begin{align*}
	\frac{\partial f} {\partial q_i} 
	&= \prod_{j \neq i} q_j (n-1-\sum_{j\neq i} q_j) -2\cdot \prod_j q_j
	\end{align*}
	By comparing it to $0$ and some rearranging we get that:
	\begin{align*}
	& 2\cdot \prod_j q_j = \prod_{j \neq i} q_j (n-1-\sum_{j\neq i} q_j) 
	\\
	&\implies 2q_i = n-1-\sum_{j\neq i} q_j 
	\\
	&\implies q_i = n-1-\sum_{j} q_j
	\end{align*}
	Thus,we have that for every $i$, $q_i$ has the same value of $q_i = n-1-\sum_{j} q_j$ and to compute the value of $q_i$ we can solve: $q = n-1 -(n-1) q$ which implies that $q = \frac{n-1}{n}$. Thus, we have that in our original maximization problem, for every $i$, $p_i = 1/n$. 
\end{proof}

\begin{lemma} \label{lem-lambda-small-1}
	For any $n \geq 3$ the value of $\lambda$ in Theorem \ref{thm-fan-tight-bound} is smaller than 1. 
	%For theorem \ref{thm-fan-tight-bound}, $\lambda < 1$.
\end{lemma}
\begin{proof}
	Recall that $\lambda = \frac{\pi_o}{(1-p)^{n-1}(n-1)c}$	where $\pi_o = \frac{1 - (1-p)^{n-1}}{p}(p-c)$, $p = \frac{1}{n}$ and $c = \frac{1}{n} - \frac{1}{n^2}$. By plugging in the values of $p, c$ and $\pi_o$ we get that
	\begin{align*}
	\lambda =  \frac{(1 - (1 - \frac{1}{n})^{n-1})}{n(1-\frac{1}{n})^{n+1}}
	\end{align*}
	%	To show that $\lambda \leq 1$ it is sufficient to prove that
	%	\small
	%	\begin{align*}
	% \frac{1}{n} - \frac{1}{n}(1 - \frac{1}{n})^{n-1} - (1 - \frac{1}{n})^{n+1} < 0
	%	\end{align*}
	%	\normalsize
	%	Let $f(x) = \frac{1}{x} - \frac{1}{x}(1 - \frac{1}{x})^{x-1} - (1 - \frac{1}{x})^{x+1}$. We prove that $f(x) < 0$ for $x \geq 4$ which completes the proof. First, observe that $f(4) = -95/1024 < 0$. We now prove that $f'(x) < 0$ for $x \geq 4$ and therefore $f(x) < f(4) < 0$ for any $x \geq 4$. Note that 
	%	\small
	%	\begin{align*}
	%	f(x) &= \frac{1}{x} - \frac{1}{x}(1 - \frac{1}{x})^{x-1} - (1 - \frac{1}{x})^{x+1} \\
	%	&= \frac{1}{x} - (1 - \frac{1}{x})^{x - 1}\Big[\frac{1}{x} + (1 - \frac{1}{x})^{2}\Big] \\
	%	&= \frac{1}{x} - (1 - \frac{1}{x})^{x - 1}\Big[\frac{1}{x^2} - \frac{1}{x} + 1\Big]
	%	\end{align*}
	%	\normalsize
	%	
	%	
	%	\normalsize
	To show that $\lambda\leq 1$ it suffices to show that:
	\begin{align*}
	(1 - \frac{1}{n})^{n-1} + n(1-\frac{1}{n})^{n+1} \geq 1
	\end{align*}
	
	%	It is sufficient to prove that
	%	\small
	%	\begin{align*}
	%		& \frac{1}{n}(1 - (1 - \frac{1}{n})^{n-1}) < (1-\frac{1}{n})^{n+1} \iff \\
	%		& \frac{1}{n} - \frac{1}{n}(1 - \frac{1}{n})^{n-1} - (1 - \frac{1}{n})^{n+1} < 0
	%	\end{align*}
	%	\normalsize
	Let $f(n) = (1 - \frac{1}{n})^{n-1} + n(1-\frac{1}{n})^{n+1} $. Observe that $f(3)=28/27>1$. Thus, showing that $f(n)$ is increasing will complete the proof.
	%	We denote $f(x) = \frac{1}{x} - \frac{1}{x}(1 - \frac{1}{x})^{x-1} - (1 - \frac{1}{x})^{x+1}$ and prove that $f(x) < 0$ for $x \geq 4$ which completes the proof. First, observe that $f(4) = -95/1024 < 0$. We now prove that $f'(x) < 0$ for $x \geq 4$ and therefore $f(x) < f(4) < 0$ for any $x \geq 4$. 
	%	Note that 
	%	\small
	%	\begin{align*}
	%	f(x) &= \frac{1}{x} - \frac{1}{x}(1 - \frac{1}{x})^{x-1} - (1 - \frac{1}{x})^{x+1} \\
	%	&= \frac{1}{x} - (1 - \frac{1}{x})^{x - 1}\Big[\frac{1}{x} + (1 - \frac{1}{x})^{2}\Big] \\
	%	&= \frac{1}{x} - (1 - \frac{1}{x})^{x - 1}\Big[\frac{1}{x^2} - \frac{1}{x} + 1\Big]
	%	\end{align*}
	%	\normalsize
	Note that 
	\begin{align*}
	f'(n) &= \frac{\left(\frac{n-1}{n}\right)^n \left(\left(n^2-n+1\right) \ln \left(\frac{n-1}{n}\right)+2 n-1\right)}{n-1}
	\end{align*}
	and by using calculus one can show that it is indeed the case that $f'(n)>0$ for any $n>2$ which completes the proof. 
\end{proof}

\section{Three Node Instances} \label{app-3-nodes}

	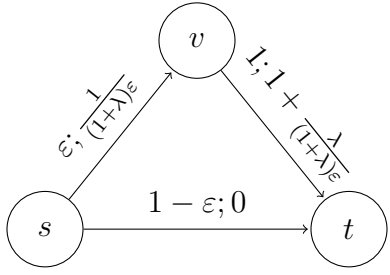
\begin{figure}[]
	\centering
	\begin{tikzpicture}[->,shorten >=1pt,auto,node distance=2cm, thin]
	\node(0) [circle, draw, minimum size=1cm] at (0,0) {$s$};
	\node (1) [circle, draw, minimum size=1cm] at (2,2.5) {$v$};
	\node (5) [circle, draw, minimum size=1cm] at (4,0) {$t$};
	
	\path[every node/.style={sloped,anchor=south,auto=false}]
	(0) edge node {$\eps; \frac{1}{(1+\lambda)\eps}$} (1)	
	(0) edge node {$1-\eps; 0$} (5)
	(1) edge node {$1; 1+ \frac{\lambda}{(1+\lambda)\eps}$} (5)
	;
	\end{tikzpicture}
	\caption{On each edge edge the left expression is the probability of taking the edge and the right number is the cost if the edge is taken. For $R=1$, we have that $\pi_s = 0$ and $\pi_o = \frac{\lambda}{1+\lambda} \cdot R -\eps$.  }
	\label{fig-alt-model}
\end{figure}

\begin{claim} \label{clm:alternative}
	In an alternative model in which costs are positioned on the the edges. For any $\eps$ there exists a $3$-node graph in which $\pi_s = 0$ and $\pi_o = \frac{\lambda}{1+\lambda} \cdot R -\eps$.
\end{claim}
\begin{proof}
	Consider the $3$-node graph depicted in Figure \ref{fig-alt-model}. In this graph, it is clear that the optimal agent will not continue from $v$ to $t$ but the sophisticated agent will. Thus, the expected payoff of the optimal agent is:
	\begin{align*}
	\pi_o = (1-\eps)\cdot 1 -\eps \cdot \frac{1}{(1+\lambda)\eps} =  \frac{\lambda}{1+\lambda} - \eps
	\end{align*}
	If a $\lambda$ biased sophisticated agent will choose to traverse the graph he will always reach the target. Thus, its expected payoff will be:
	\begin{align*}
	1-\eps(\frac{1}{(1+\lambda)\eps} + \frac{\lambda}{(1+\lambda)\eps}) = 0 
	\end{align*}
	Thus, the payoff of the sophisticated agent is $0$.
	
\end{proof}

	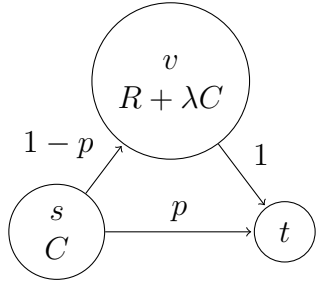
\begin{figure}[] \label{fig-3-tight}
	\centering
	\begin{tikzpicture}[->,shorten >=1pt,auto,node distance=2cm, thin]
	\node (s) [circle, draw, inner sep=0pt,minimum size=1pt] at (0,0) {\begin{tabular}{c} $s$
		\\ $C$ \end{tabular}};
	\node (v) [circle, draw,inner sep=0pt,minimum size=1pt] at (1.5,2) {\begin{tabular}{c} $v$
		\\ $R+\lambda C$ \end{tabular}};
	\node (t) [circle, draw, inner sep=0pt,minimum size=0.8cm] at (3,0) {$t$};
	
	\path
	(s) edge node {$1-p$} (v)
	(s) edge node {$p$} (t)
	(v) edge node {$1$} (t)
	;
	\end{tikzpicture}
	\caption{$3$-node graph illustration for Claim \ref{clm-3-nodes}.} \label{fig:tne}
\end{figure}

\begin{claim} \label{clm-3-nodes}
	For any 3-node graph and any $\lambda\geq 0$, $\pi_s \geq \pi_o - \frac{2+\lambda -2\sqrt{1+\lambda}}{\lambda} \cdot R$ and this is tight.
\end{claim}
\begin{proof}
	Consider the graph in figure \ref{fig:tne}. {First, one can observe that the only two possible scenarios in which the optimal and the sophisticated agent will have different payoffs are:
	\begin{itemize}
		\item The optimal agent traverses the graph for a single step and the sophisticated agent continues.
		\item The optimal agent traverses the graph for a single step and the sophisticated agent is unwilling to start.
	\end{itemize}
	These are the only scenarios we should consider as if the optimal agent does not traverse the graph or continues at $v$ then the sophisticated agent will do the same.}

%	We will see that the payoff difference is maximized in case that the sophisticated agent does not traverse the graph. 
	{
	We begin by considering the scenario in which the optimal agent traverses the graph for a single step and the sophisticated agent continues. Observe that $\pi_o = pR - C$. Notice that if the sophisticated agent would start traversing the graph it would continue at $v$. Thus, its expected payoff for traversing the graph is $pR - C - (1-p)\lambda C \le 0$. By rearranging we get that $	p \le \frac{C(1+\lambda)}{R + \lambda C}$.
}
%	\begin{align*}
%	pR - C - (1-p)\lambda C &\le 0 \\
%	pR + \lambda p C &\le C + \lambda C \\
%	p &\le \frac{C(1+\lambda)}{R + \lambda C}
%	\end{align*}
	Thus, to maximize the expected payoff of the optimal agent, we set $p =
	\frac{C(1+\lambda)}{R+\lambda C}$ and get that:
	\begin{align*}
	\pi_o &= \frac{C(1+\lambda)}{R+\lambda C} \cdot R - C
	\end{align*}
	To maximize $\pi_o$ we take a derivative with respect to $C$ and compare it to $0$:
	\begin{align*}
		\frac{\partial\pi_o}{\partial C} = \frac{R(1+ \lambda)(R + \lambda C) - \lambda CR(1 + \lambda)}{(R + \lambda C)^2} - 1 = \frac{R^2(1+ \lambda)}{(R + \lambda C)^2} - 1
	\end{align*}
	\begin{align*}
		\frac{R^2(1+ \lambda)}{(R + \lambda C)^2} - 1 = 0 \implies R^2(1+\lambda) = (R+\lambda C)^2
	\end{align*}
	\begin{align*}
		C = \frac{R}{\lambda}\p{\sqrt{1+\lambda}-1}
	\end{align*}
	which gives
	\begin{align*}
	p =
	\frac{R\p{\sqrt{1+\lambda}-1}(1+\lambda)}{\lambda(R+R(\sqrt{1+\lambda}-1))}
	&= \frac{\p{\sqrt{1+\lambda}-1}\sqrt{1+\lambda}}{\lambda} 
	\\
	&= \frac{1+\lambda-\sqrt{1+\lambda}}{\lambda}
	\end{align*}
	Therefore:
	\begin{align*}
		\pi_o &= \frac{1+\lambda-\sqrt{1+\lambda}}{\lambda} \cdot R - \frac{R}{\lambda}(\sqrt{1 + \lambda} - 1)
		\\
		&= \frac{(2 + \lambda - 2\sqrt{1 + \lambda})}{\lambda} \cdot R
	\end{align*}
	%For $0 < \lambda \le 1$, $\frac{1}{2} < p \le .586$.
	For $0 < \lambda \leq 1$ we get that $0 < \pi_o - \pi_s \leq 0.172R$. 
	
	%This was removed to the main text
	%Note that for any $\lambda\geq 0$, $\frac{\lambda}{\lambda + 1} \geq \frac{2 + \lambda - 2\sqrt{1 + \lambda}}{\lambda}$ which implies that the bound we computed in Proposition \ref{prop:not-tight-bound-soph} is not tight.

	Finally, we consider the scenario in which the sophisticated agent traverse the graph and continues at $v$ while the optimal agent stops traversing the graph at $v$. { We show that optimizing the payoff difference for this scenario get us to the same optimization problem as we just solved.} Denote by $c(v)$ the cost at $v$. Since the sophisticated agent continues we have that $R-c(v) \geq -\lambda C \implies c(v) \leq R+\lambda C$. Also,  since the expected payoff of the sophisticated agent is positive we have that:
	\begin{align*}
		\pi_s = pR - C + (1-p)(R-c(v)) &> 0 \implies 
		\\R + pc(v) -c(v) -C &> 0 \implies 
		\\ c(v) &< \frac{R - C}{1 - p}
	\end{align*}
	Consider the difference between the payoffs of the agents:
	\begin{align*}
	\pi_o-\pi_s= -(1-p)(R-c(v)) = (1-p)(c(v)-R) 
	\end{align*}
	Clearly, the difference is maximized for the maximal value of $c(v)$. Since 
	$c(v) \leq \min\{\frac{R - C}{1 - p},R+\lambda C\}$ we get that this value is maximized
	when $\frac{R - C}{1 - p} = R+\lambda C$ by rearranging we get that in this case
	$p=\frac{C(1+\lambda)}{R + \lambda C}$. Since in this case we have that:
	\begin{align*}
	\pi_o-\pi_s \leq  (1-p)(\min\{\frac{R - C}{1 - p},R+\lambda C\} - R) = p\cdot R - C
	\end{align*}
This implies the exact optimization problem as in the first case, which completes the proof.
\end{proof}
%\soedit{
%\begin{corollary}
% For any value of $\lambda\geq0$ there exists a 3-node instance such that, $\pi_s = \pi_o - \frac{(2+\lambda -2\sqrt{1+\lambda})}{\lambda} \cdot R$.
%\end{corollary}
%\begin{proof}
%Consider setting for the instances illustrated in Figure \ref{fig-3-tight}, $	C = \frac{R}{\lambda}\p{\sqrt{1+\lambda}-1}$ and 
%\end{proof}
%
%}

\end{document}